\documentclass[11pt]{article}
\usepackage{amsmath, amssymb, amsthm}
\usepackage{amsfonts}
\usepackage{upgreek}
\usepackage{bm}
\usepackage[margin=1in]{geometry}
\usepackage{enumerate}
\usepackage[shortlabels]{enumitem}
\usepackage{booktabs}
\usepackage{etoolbox}
\usepackage{mathtools}
\usepackage{graphicx}
\usepackage{titlesec}
\usepackage{braket}
\usepackage{float}
\usepackage{nicefrac}
\usepackage{xcolor,tikz}
\usetikzlibrary{arrows.meta,positioning}
\usepackage[multiple]{footmisc}
\usepackage{float}
\usepackage{bbm}
\usepackage{dsfont}
\usepackage{mathrsfs}
\usepackage{adjustbox}
\usepackage{multirow}
\usepackage{makecell}
\usepackage[math]{cellspace}
\usepackage{hyperref}
\usepackage{xcolor}
\hypersetup{
    colorlinks=true,
    linkcolor = {blue},
    citecolor = {teal},
    urlcolor = {blue}
}

\newcommand{\R}{{\mathbb R}}

\newcommand{\N}{{\mathbb N}}

\newcommand{\E}{{\mathbb E}}

\DeclareMathOperator{\conv}{conv}

\newcommand{\TC}{\mathcal T}

\newcommand{\norm}[1]{\|#1\|}

\newcommand{\indic}{\mathds{1}}

\newcommand{\calO}{{\mathcal{O}}}

\newcommand{\opt}{{\textsc{Opt}}}

\newcommand{\alg}{{\textsc{Alg}}}

\AtBeginEnvironment{theorem}{\setlist{before=\leavevmode, nosep}}

\usepackage{amsmath, amssymb, amsthm}
\usepackage[numbers]{natbib}
\usepackage{cleveref}
\usepackage{algorithm}
\usepackage{algpseudocode}
\usepackage{graphicx}
\usepackage{datetime}
\newtheorem{theorem}{Theorem}
\newtheorem{lemma}[theorem]{Lemma}
\newtheorem{corollary}[theorem]{Corollary}
\newtheorem{proposition}[theorem]{Proposition}

\providecommand{\proofname}{Proof}
\newcommand{\jmlrBlackBox}{\rule{1.5ex}{1.5ex}}
\newcommand{\jmlrQED}{\hfill\jmlrBlackBox\par\bigskip}
\renewenvironment{proof}[1][]{
    \ifx\newenvironment#1\newenvironment
      \par\noindent{\bfseries\upshape \proofname\ }
    \else
      \def\temp{#1}%
      \def\noname{noname}%
      \ifx\temp\noname
      \else
        \par\noindent{\bfseries\upshape #1\ }
      \fi
    \fi
    \ignorespaces
  }{%
    \jmlrQED
  }

\crefname{lemma}{Lemma}{lemmas}
\crefname{proposition}{Proposition}{propositions}
\crefname{remark}{Remark}{remarks}
\crefname{corollary}{Corollary}{corollaries}
\crefname{definition}{Definition}{definitions}

\newdateformat{monthyeardate}{%
  \monthname[\THEMONTH] \THEYEAR}

\begin{document}
\title{A Linear Lower Bound for Convex Body Chasing in Every Norm}

\author{
  Nicolas Christianson\thanks{Department of Computer Science, Johns Hopkins University. Email: \href{mailto:christianson@jhu.edu}{\nolinkurl{christianson@jhu.edu}}.}
}

\date{\monthyeardate\today}

\maketitle

\begin{abstract}%
    We establish an \(\Omega(d)\) lower bound on the competitive ratio of convex body chasing in $\R^d$ with movement costs measured in any norm. Previously, the best known lower bound valid across all norms was $\Omega(d^{1/6})$, established by Papazov~\cite{papazov2023role}; stronger lower bounds were known for $\ell_p$ norms, where Bubeck, Klartag, Lee, Li, and Sellke had established a lower bound of $\Omega(\max\{\sqrt{d}, d^{1-1/p}\})$~\cite{bubeck_chasing_2019}. In particular, a linear lower bound was known only for the $\ell_\infty$ norm, realized by the standard construction of chasing faces of a hypercube. Our lower bound strictly improves on the prior universal and $\ell_p$ lower bounds for $p < \infty$ and shows that the functional Steiner point algorithm of Sellke~\cite{sellke_chasing_2023}---which has competitive ratio $\calO(d)$---is asymptotically optimal for every norm. Our approach is based on the recursive construction of Bubeck, Coester, and Rabani~\cite{bubeck_randomized_2022} for metrical service systems, which gives a lower bound of $\Omega(r^2)$ for requests of size $\calO(r)$. We show that requests in their construction, each comprising a set of discrete points in a graph metric, can be mapped to the vertices of convex requests in a $(2r+1)$-dimensional $\ell_p$ vector space while approximately preserving movement costs up to a factor of $\calO(r)$; the general result then follows via Milman's quotient-of-a-subspace theorem~\cite{milman_1985}.
\end{abstract}

\section{Introduction}

We consider the problem of \emph{convex body chasing} in $\R^d$ endowed with an arbitrary norm $\|\cdot\|$: starting from a point $\mathbf{x}_0 \in \R^d$, an algorithm faces a sequence of nonempty closed convex sets $K_1, \ldots, K_T \subseteq \R^d$ that are revealed online, where $T \in \N$ is unknown a priori. After seeing body $K_t$, the algorithm must choose a point $\mathbf{x}_t \in K_t$, and pays the cost $\|\mathbf{x}_t - \mathbf{x}_{t-1}\|$ of moving from its previous decision. The algorithm's total cost is its total movement $\sum_{t=1}^T \|\mathbf{x}_t - \mathbf{x}_{t-1}\|$.

For an online algorithm $\alg$, let $\alg(K_1, \ldots, K_T)$ denote its cost on the instance consisting of bodies $K_1, \ldots, K_T$ (or expected cost, if $\alg$ is randomized), and let $\opt(K_1, \ldots, K_T)$ be the cost of the \emph{offline optimal} algorithm that knows all bodies in advance and can optimize its movement accordingly. The standard goal is to design an algorithm with the smallest possible competitive ratio:

$$\mathrm{CR}(\alg) \coloneqq \sup_{K_1, \ldots, K_T; T \in \N} \frac{\alg(K_1, \ldots, K_T)}{\opt(K_1, \ldots, K_T)}.$$

Convex body chasing was first introduced by Friedman and Linial~\cite{friedman_convex_1993}, who focused on the Euclidean ($p=2$) setting, establishing a lower bound of $\sqrt{d}$ on the competitive ratio of any algorithm and conjecturing that a finite competitive ratio (dependent on $d$) was possible.
After several decades of progress, \cite{bubeck2019competitively} resolved this conjecture, giving an upper bound of $2^{\calO(d)}$. Subsequent work of \cite{argue_chasing_2021,sellke_chasing_2023} improved the upper bound to $O(\min\{d, \sqrt{d \log T}\})$ in Euclidean space. In particular, the \emph{functional Steiner point} algorithm of \cite{sellke_chasing_2023} has competitive ratio $d$ for \emph{any} norm. These results also imply an $\calO(d)$ competitive ratio for convex \emph{function} chasing, where instead of receiving a body $K_t$ and the hard constraint of containment at each time, the algorithm receives a convex cost function $f_t : \R^d \to \R_+$, and pays the cost of its decision $f_t(\mathbf{x}_t)$, in addition to movement cost \cite{kranakis_chasing_2016,bubeck2019competitively}.

On the other hand, \cite{bubeck_chasing_2019} established a lower bound of $\Omega(\max\{\sqrt{d}, d^{1-1/p}\})$ for the $\ell_p$ norm with any $p \in [1, \infty]$ by combining the ``chasing faces of a hypercube'' construction of \cite{friedman_convex_1993} with a Hadamard matrix construction. Their construction is a \emph{nested} instance (meaning each body is contained in the previous one), and they also provide an algorithm for the nested setting which obtains a matching upper bound (up to log factors) of $\tilde{\calO}(\max\{\sqrt{d}, d^{1-1/p}\})$. Their lower bound is linear exactly when $p = \infty$, and shows that Sellke's functional Steiner point algorithm is optimal for general (not-necessarily-nested) convex body chasing with the $\ell_\infty$ norm \cite{sellke_chasing_2023}. More generally, the best universal lower bound across norms is $\Omega(d^{1/6})$, established by \cite{papazov2023role} by extending the hypercube construction via a bound on the Banach-Mazur distance to the $\ell_\infty$ ball. Thus there remains a polynomial gap between the best known upper and lower bounds on the competitive ratio of convex body chasing---up to $\Theta(d^{5/6})$ for general norms, and $\Theta(\sqrt{d})$ for $\ell_p$ norms with $p \in [1, 2]$, which includes the Euclidean case originally considered by Friedman and Linial~\cite{friedman_convex_1993}---leaving open the question of whether convex body chasing with a sublinear competitive ratio is possible in any norm.

In this paper, we resolve this gap, showing that the competitive ratio for general convex body chasing is $\Omega(d)$ for every norm:

\begin{theorem}\label{thm:main}
There is a universal constant \(c>0\) such that, for every \(d\ge 2\)
and every norm $\|\cdot\|$ on $\R^d$, any algorithm for convex
body chasing in $(\R^d, \|\cdot\|)$ has competitive ratio at
least $cd$.
\end{theorem}

Notably, this implies that Sellke's functional Steiner point algorithm~\cite{sellke_chasing_2023}, which has competitive ratio $\calO(d)$, is optimal (up to a constant factor) for \emph{every} norm; this also resolves in the negative the question of whether a sublinear competitive ratio is possible in any norm (such as the Euclidean norm). Moreover, this linear lower bound stands in stark contrast to the nested case, where a sublinear competitive ratio is possible for every $\ell_p$ norm with $p \in [1, \infty)$; thus we establish that general convex body chasing is, in general, \emph{strictly} harder in a competitive sense than nested chasing.

\paragraph{Overview of the proof.}
Our starting point is the recursive lower bound construction of Bubeck, Coester, and Rabani~\cite{bubeck_randomized_2022} for metrical service systems, and specifically the construction with ``small requests'' described in Section 3.6 of their paper (which also applies to layered graph traversal). They give a recursively defined family of graphs indexed by $r \in \{0, 1, 2, \ldots \}$ where, for the graph at level $r$, the requests contain at most $r + 1$ vertices of the graph; the offline optimal strategy can follow the shortest path through the requested vertices, while any (possibly randomized) online algorithm must pay a multiplicative factor $\Omega(r^2)$ more in expectation.
In \Cref{sec:recursive} we describe their construction and relevant results that we will use in our reduction.

Our construction for convex body chasing will replace each discrete request by the convex hull of its image
under a map \(\Psi\) into $\R^D$ with $D = 2r + 1$.
Convexity gives the algorithm additional flexibility relative to the original discrete instance: a point in a convex hull
can be any weighted average of the mapped vertices. In particular, small movement of that average does not necessarily imply that the corresponding weights can be moved cheaply between vertices in the original graph.
Our map is designed to avoid this problem on the requests in the construction: we show that the distance between two weighted averages controls the transportation cost between the weights, which allows us to relate the cost of a convex body chasing algorithm to that of a randomized algorithm for the original discrete requests using the former's weights.

The coordinates of the map $\Psi$ have a simple geometric interpretation. Each recursive
graph has a source and a target, and all paths that move forward from
the source to the target have the same length. The first coordinate
of \(\Psi(v)\) is the distance traveled from the overall source to \(v\) along
such a path; we call this its \emph{depth}. However, depth alone cannot
distinguish vertices on different branches and at different levels of recursion. We therefore add two coordinates for each level of recursion, one for each branch. On vertices in its branch, a coordinate equals the distance of that vertex to the nearer endpoint of the graph copy it resides in, and is zero elsewhere. In particular, these coordinates return to zero at the end of a graph copy at a given level, so the same pair of coordinates can be used again to simulate the next graph copy. See \Cref{sec:map} for further details.

This construction allows us to relate the movement cost of a convex body chasing algorithm to the transportation cost of a corresponding randomized algorithm for the discrete request sequence (Lemma~\ref{lem:transport}). Then, in Section~\ref{sec:convexification}, we use this comparison to invoke the \cite{bubeck_randomized_2022} lower bound, losing only a factor \(8D^{1-1/p}\) on the online cost and a factor \(D^{1/p}\) on the offline cost. The product of these factors is \(8D\), which implies the desired $\Omega(d)$ lower bound for $\ell_p$ movement costs with any $p \in [1, \infty]$, since we can choose $r = \Theta(d)$. To extend the result to arbitrary norms, we invoke Milman's quotient-of-a-subspace theorem~\cite{milman_1985}, which allows transferring the Euclidean construction into a subspace of a general normed vector space with only a constant-factor loss on the lower bound.

\paragraph{AI use disclosure.}
The lower bound presented in this paper was obtained with the assistance of OpenAI's GPT-6 Pro model. See Section~\ref{sec:ack} for further details on the role of this model in developing the results and drafting the manuscript.

\section{An overview of the construction in \cite{bubeck_randomized_2022}}\label{sec:recursive}

We first give an overview of the lower bound for metrical service systems (which can be viewed as chasing finite sets in a metric space) established by Bubeck, Coester, and Rabani~\cite{bubeck_randomized_2022}; the construction and results detailed in this section are from Sections 3.1--3.6 of their paper (in particular, the construction is the ``small set'' construction described in Section 3.6 of their paper), and we include them for completeness. Then, in the next section, we describe how their construction is mapped to a convex body chasing instance.

In the metrical service systems problem with size $r \in \N$, we have a metric space $(M, d)$ and a starting point $x_0 \in M$. At each time, an adversary requests a nonempty subset of points $\rho_t \subseteq M$ with cardinality $|\rho_t| \leq r$. In response, the algorithm must choose a point $x_t \in \rho_t$, and incurs cost equal to the distance $d(x_t, x_{t-1})$ from its previous point. If the algorithm is randomized, its decision constitutes a probability distribution over the points in $\rho_t$, and consecutive distributions can be coupled so that it pays movement (in expectation) equal to the optimal transportation distance between the distributions.

\cite{bubeck_randomized_2022} constructs a recursive family of metric spaces defined by finite undirected weighted graphs $G$, where vertices of the graph correspond to points in the metric, and the metric distance corresponds to the shortest-path distance between two vertices; we write this distance as $d_G(u, v)$ for vertices $u$ and $v$. For a vertex $v$ of $G$, we will abuse notation and write $v \in G$. They also construct a distribution over request sequences for which the online algorithm's
expected cost grows quadratically with the recursion level, relative
to the offline optimum. We describe their construction, focusing on
the modification in Section 3.6 of their paper that keeps requests small (of size proportional to the recursion level), and quote
their lower bound. We also record the geometric properties needed for our reduction. Then, in Section~\ref{sec:map}, we will describe how these instances are mapped to a vector space to enable our lower bound on convex body chasing.

\subsection{The recursive family of graphs}\label{subsec:graphs}

The construction in \cite[Section 3.6]{bubeck_randomized_2022} defines a recursive family of graphs \(G_0,G_1,\ldots\), each with a designated source and target vertex; each sequence of requests constituting a problem instance will start at the source and terminate at the target. Fix a sequence of natural numbers $m_0 \leq m_1 \leq m_2 \leq \cdots$ and real numbers $C_r \geq 1$ for $r \in \{0, 1, 2, \ldots\}$, with $m_r \geq C_r$ and $C_0 = 1$. The parameter $m_r$ will determine the number of copies of $G_r$ used to construct $G_{r+1}$, and $C_r$ will denote a lower bound on the competitive ratio, determined later by induction; we index with the variable $r$, which we call the \emph{level} of the recursion.
The base graph $G_0$ is a single edge of length one, with source vertex $s_0$ and target vertex $t_0$.

Proceeding with the induction, suppose that $G_r$ has been constructed; in the following, when we describe ``scaling'' $G_r$, this is achieved by homogeneously scaling its edge lengths to achieve a specified source-to-target distance. We construct $G_{r+1}$ by specifying a source vertex $s_{r+1}$ and target vertex $t_{r+1}$, and connecting them via two parallel branches composed of scaled copies of $G_r$. Specifically, each of the two branches of $G_{r+1}$ consists of, in order,

\begin{enumerate}
 \item an entrance edge of length $b_r \coloneqq \frac{C_r}{m_r} \in [\frac{1}{m_r}, 1]$;
 \item \(m_r^2\) normalized copies of \(G_r\), each scaled to have
       source-to-target distance \(1/m_r\);
 \item \(2m_r\) normalized copies of \(G_r\), each scaled to have
       source-to-target distance one.
\end{enumerate}
See Figure~\ref{fig:block} for an illustration. The end of each entrance edge is identified with the source of the first copy of $G_r$ along its branch. We then join consecutive copies of $G_r$ along each branch by identifying the target vertex of one with the source vertex of the next. The target vertices of both of the final copies of $G_r$ on each branch are identified with the target $t_{r+1}$. The two branches share only the level-$(r+1)$ source $s_{r+1}$ and target $t_{r+1}$ as vertices; each branch thus has length
\begin{equation}
 L_{r+1}=3m_r+b_r,
 \label{eq:parent-length}
\end{equation}
which is equivalently the level-$(r+1)$ source-to-target distance $d_{G_{r+1}}(s_{r+1}, t_{r+1})$.

We will refer to the copies of $G_r$ placed within $G_{r+1}$ as ``small copies'' if they are normalized to have length $1/m_r$, and ``large copies'' if they are normalized to have length one. Note that these lengths are defined relative to the scale of the level-$(r+1)$ graph $G_{r+1}$: when this graph is used to build $G_{r+2}$, it will also be scaled accordingly (and thus all of the copies of $G_r, G_{r-1}, \ldots$ inside of $G_{r+1}$ will be scaled down further).

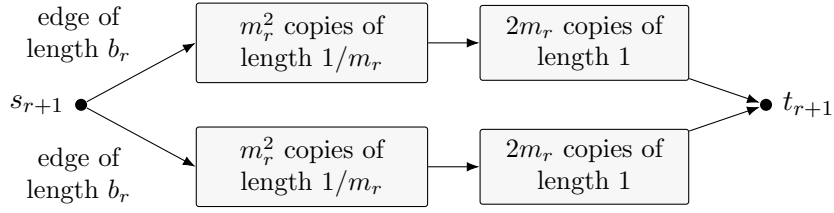
\begin{figure}[t]
\centering
\begin{tikzpicture}[x=1cm,y=1cm,>=Latex,
  copy/.style={draw,rounded corners=1pt,minimum height=.62cm,align=center,
    font=\small,fill=black!3},
  point/.style={circle,fill,inner sep=1.6pt}]
 \node[point,label=left:\(s_{r+1}\)] (s) at (0,0) {};
 \node[copy,text width=2.8cm] (sl) at (3.05,.82)
       {\(m_r^2\) copies of\\length \(1/m_r\)};
 \node[copy,text width=2.8cm] (sr) at (3.05,-.82)
       {\(m_r^2\) copies of\\length \(1/m_r\)};
 \node[copy,text width=2.5cm] (ll) at (6.65,.82)
       {\(2m_r\) copies of\\length \(1\)};
 \node[copy,text width=2.5cm] (lr) at (6.65,-.82)
       {\(2m_r\) copies of\\length \(1\)};
 \node[point,label=right:\(t_{r+1}\)] (t) at (9.05,0) {};
 \draw[->] (s)--node[above left,align=center,font=\small]{edge of\\length \(b_r\)}(sl.west);
 \draw[->] (s)--node[below left,align=center,font=\small]{edge of\\length \(b_r\)}(sr.west);
 \draw[->] (sl)--(ll); \draw[->] (sr)--(lr);
 \draw[->] (ll)--(t); \draw[->] (lr)--(t);
\end{tikzpicture}
\caption{The graph \(G_{r+1}\) in the small-request construction of \cite[Section 3.6]{bubeck_randomized_2022}. Each box is a chain of copies of
\(G_r\), with the target of one copy serving as the source of the
next; the displayed length is that of an individual copy. The requests
first alternate through the small copies, then advance a randomly
chosen branch through one large copy in each of \(m_r\) rounds, and
finally finish the branch that advanced fewer times. Arrows indicate
the direction used to define depth, but an algorithm may move
along edges in either direction.}
\label{fig:block}
\end{figure}

\paragraph{Paths and depth.}
Although the graphs are constructed with (and movement costs use) undirected edges and distances, it is useful to
orient edges from source to target at each level of the recursion to reason about instances.
A directed path from source $s_{r+1}$ to target $t_{r+1}$ through \(G_{r+1}\) traverses the entrance edge and
every copy of $G_r$ on one branch; by symmetry, regardless of the choices made within each copy, the resulting path length is $L_{r+1} = d_{G_{r+1}}(s_{r+1}, t_{r+1})$.
Every vertex in $G_{r+1}$ lies on such a path; the same reasoning applies to vertices in graphs of all levels of recursion.

Let $H$ be a copy of the graph $G_r$ within some graph $G$ at a higher level of recursion (say, $G = G_{r+1}$); when discussing such a copy $H$, we will consider distances with respect to the higher-level graph $G$ it resides in (i.e., with the scaling of $G_r$ used to construct $G$). For a vertex $v \in H$, we define the \emph{local depth} $h_H(v)$ of $v$ to be the length of any directed path from the source $s_H$ to $v$---i.e., $h_H(v) = d_G(s_H, v)$; this is well-defined by symmetry. We write $h(v)$ for the depth of a vertex measured from the source $s_G$ of the full graph. Clearly we can decompose the depth of a vertex $v \in H$ in terms of the depth of the source of $H$ and the local depth of $v$: for any $v \in H$, $h(v) = h(s_H) + h_H(v)$.

Depth increases by exactly the edge length along each directed edge,
so its absolute change along any undirected path is at most that
path's length. Thus $h$ is $1$-Lipschitz in the full graph
distance: for $u, v \in G$,
$$|h(u) - h(v)| \leq d_G(u, v).$$
Moreover, for a copy $H$, $h_H$ is likewise $1$-Lipschitz in the full graph distance.
If its local source-to-target length is $L_H$, we have both
\[
 d_G(s_H,v)=h_H(v) \text{ and }
 d_G(v,t_H)=L_H-h_H(v)
\]
for $v \in H$. The triangle inequality then implies that, for two vertices $u, v \in H$,
\begin{equation}
 d_G(u,v)\le
 \min\{h_H(u)+h_H(v),\,2L_H-h_H(u)-h_H(v)\}
 \le L_H.
 \label{eq:copy-distance-bound}
\end{equation}
That is, the diameter of any copy $H$ is exactly $L_H$.

\subsection{The request sequence}\label{subsec:requests}

We now describe the random request sequence on each graph, following
\cite[Sections 3.2 and 3.6]{bubeck_randomized_2022}; recall that at each time, the request specifies a set of points that are feasible for the algorithm.

The request sequence on \(G_0\) simply consists of one request to its source \(\{s_0\}\), followed by a request to its target \(\{t_0\}\). At each
higher level of the recursion, requests progress through the graph copies on the two branches, first alternating between them and then choosing branches at random to progress. Finally, whichever branch advanced a greater number of times is made infeasible for the rest of the sequence, and the requests continue along the surviving branch to the target. An offline algorithm can choose this surviving branch from the start, while an online algorithm must choose whether to progress before all random choices are known.

To make the recursion precise, suppose that the request sequence on \(G_r\) has
been specified, beginning at \(\{s_r\}\) and ending at \(\{t_r\}\); we will define the level-$(r+1)$ request sequence on \(G_{r+1}\).
We first request its source
$\{s_{r+1}\}$, then we simultaneously request the pair of vertices at the ends of the two entrance edges. For each branch, we will keep track of its \emph{current vertex}, which will be initialized as this terminal vertex at the end of the entrance edge (equivalently, the source vertex of the first small copy of $G_r$ along that branch). In the following, we will distinguish the two branches as ``left'' and ``right.''

In the rest of the sequence, ``advancing a branch through its next copy'' means that, if $H$ is the next copy of $G_r$ along a branch---i.e., the copy for which the current vertex is the source---we issue the level-$r$ request sequence in $H$ using fresh, independent randomness; we call this a \emph{call} in $H$. If $z$ is the current vertex on the opposite branch (and if it has not been removed---i.e., if we haven't yet reached Stage 3 below), we also include $z$ as a feasible vertex in every request of the call: any level-$r$ request $U \subseteq H$ in a call to $H$ becomes a level-$(r+1)$ request $U \cup \{z\}$.
Thus the algorithm may remain at $z$ while requests advance through
$H$; we call $z$ a \emph{waiting vertex}. The call begins
with $\{s_H,z\}$ and ends with $\{t_H,z\}$, at which point
the advancing branch's current vertex becomes $t_H$. This rule
allows us to specify the rest of the sequence in three stages.

\paragraph{Stage 1: alternate through the small copies.}
For \(j=1,\ldots,m_r^2\), advance through the \(j\)th small copy of $G_r$
on the left and then the \(j\)th small copy on the right, retaining
the other branch's current vertex as a feasible waiting vertex throughout each call. Each small
copy has length \(1/m_r\), so the two current vertices never differ
in depth by more than \(1/m_r\) and both finish at depth \(b_r+m_r\). At the end of this stage, the current vertex of each branch is the source vertex of the first large copy of $G_r$ on that branch.

\paragraph{Stage 2: choose large copies at random.}
For each of \(m_r\) rounds, choose a branch independently and uniformly at random, and
advance through its next large copy of $G_r$, again retaining the other branch's
current vertex as a waiting vertex. Each random branch selection is independent of all previous choices and of the randomness in the recursive calls to sub-requests.
If \(X\) is the number of times the left branch is selected, then \(X\sim\operatorname{Bin}(m_r,1/2)\), and
the branches have advanced through \(X\) and \(m_r-X\) large copies,
respectively. The final request of this stage consists of the two current vertices after the left and right branches have advanced through $X$ and $m_r - X$ large copies, respectively. The current vertex of a branch with zero advances remains at the source of its first large copy.

\paragraph{Stage 3: finish the less advanced branch.}
Request \emph{only} the current vertex on the branch that advanced fewer times in Stage 2, using a fixed rule to break a tie. We will remove the current vertex from the more advanced branch and will issue no more requests to that branch---i.e., there will be no waiting vertex on that branch.
Then advance through each of the remaining large copies of $G_r$ on the surviving branch.
The final request is $\{t_{r+1}\}$.

\medskip

Throughout a recursive call in a copy \(H\), all waiting vertices
retained at higher levels remain feasible. Let \(Z\) denote this
fixed set of waiting vertices; then each local request \(U\) in $H$ becomes a higher-level request \(U\cup Z\), with at
most one vertex in \(Z\) from each higher level. In particular,
the call to $H$ begins with \(\{s_H\}\cup Z\) and ends with
\(\{t_H\}\cup Z\). We retain both endpoint requests, including
repetitions: within each stage and between Stages 1 and 2, consecutive calls end and begin with the same request (since the target of one copy along a branch is the source of the next).
This is not the case between Stages 2 and 3, because of the removal of the more advanced waiting vertex.

Let \(\mathcal D_r\) denote the resulting distribution of level-$r$ request sequences
on \(G_r\). Every realized sequence is finite, but its length may depend
on the random branch choices.\footnote{The construction in \cite[Section 3]{bubeck_randomized_2022} pads the sequences with repeated requests so that all level-$r$ calls have the same duration; this padding is not critical for their conclusion or for our analysis.}
We will need the following properties of the construction, which were established in \cite[Section 3]{bubeck_randomized_2022}.

\begin{lemma}[Request size and offline path]\label[lemma]{lem:offline-metric}
Every request in a level-\(r\) instance has at most \(r+1\) vertices.
Every realized full sequence can be served along a directed shortest
path from its source to its target; thus the offline optimum is the
source-to-target distance $L_r$.
\end{lemma}

\begin{proof}[Proof.]
The base requests are singletons. At each recursive step, a call to a copy adds at most one waiting vertex to that copy's requests; requests outside of a call to a copy have at most two vertices. The size bound follows by induction.

To construct the offline path, choose the branch that ultimately survives the level-$r$ instance. Follow an offline directed path through each copy on that branch when it is called, staying at the waiting vertex when the other branch advances. Inducting on the copies, we obtain a feasible directed source-to-target path which, as established earlier, has length $L_r$.
\end{proof}

For a suitable choice of the parameter sequences $m_r$ and $C_r$, the construction yields the following distributional lower bound; we include a proof sketch for completeness.

\begin{lemma}[Bubeck, Coester, and Rabani, Sections 3.3-3.6~\cite{bubeck_randomized_2022}]\label[lemma]{lem:mss-lower}
Let \(L_r=d_{G_r}(s_r,t_r)\) be the offline optimal cost on the level-$r$ request family, i.e., the source-to-target distance in \(G_r\). Every (possibly randomized) online algorithm $\alg$ starting at \(s_r\) satisfies
\[
 \E[\alg(\sigma)]
 \ge \Omega(r^2)\cdot L_r,
\]
where the expectation is taken both over instances $\sigma\sim\mathcal D_r$ and the algorithm's randomness.
\end{lemma}
\begin{proof}[Proof sketch.]
First consider a deterministic algorithm; the argument proceeds by induction. Suppose $C_{r-1}$ is the competitive ratio lower bound established at level $r-1$; the base case is $C_0 = 1$. At level $r$, the offline algorithm's cost is $L_r = 3m_{r-1} + \frac{C_{r-1}}{m_{r-1}} \leq 3m_{r-1} + 1$. On the other hand, the online algorithm's expected cost is at least $\frac{C_{r-1}}{m_{r-1}}$ (for the entrance edge); plus $m_{r-1}^2 \cdot \frac{C_{r-1}}{m_{r-1}} = m_{r-1}C_{r-1}$ for progressing through Stage 1; plus $\frac{1}{2}m_{r-1} C_{r-1}$ in expectation for Stage 2 (since in each of the $m_{r-1}$ rounds of this stage, the algorithm's current branch will be chosen to advance with probability $\frac{1}{2}$, in which case it must either advance at expected cost at least $C_{r-1}$ or switch to the other branch at cost at least $2m_{r-1}$); plus $(2m_{r-1} - \E[\min\{X, m_{r-1}-X\}])C_{r-1} = \left(\frac{3}{2}m_{r-1} + \Omega(\sqrt{m_{r-1}})\right)C_{r-1}$ for Stage 3, where $X$ is the number of times the left branch was chosen. Applying the inductive bound (as in the latter three cases) when the algorithm can switch branches after observing part of a recursive call requires the strengthened induction in \cite[Section 3.5]{bubeck_randomized_2022}.
Altogether, the online algorithm's competitive ratio is at least
$$C_r \geq \frac{3m_{r-1} + \Omega(\sqrt{m_{r-1}})}{3m_{r-1} + 1}C_{r-1} = \left(1 + \Omega\left(\frac{1}{\sqrt{m_{r-1}}}\right)\right)C_{r-1}.$$
Setting $m_{r-1} = \Theta(C_{r-1})$ and solving the recursion gives the result.

The same expected-cost lower bound holds for randomized algorithms: indeed, the request distribution is fixed independently of the algorithm’s decisions and randomness, so we may fix the algorithm’s random seed, apply the deterministic lower bound, and then average over the seed. The expectation then includes both the random request sequence and the algorithm's randomness.
See \cite{bubeck_randomized_2022} for full proof details.
\end{proof}

Since every feasible trajectory through an instance has cost at least $L_r$, the lower bound in the above lemma can be stated explicitly as $\E[\alg(\sigma)] \geq c_0 r^2 L_r$ for all $r \geq 1$, for a universal constant $c_0 > 0$.
In the next section, we will describe how to port this lower bound construction into the convex body chasing setting.

\section{Mapping requests into $\R^D$}\label{sec:map}

Fix a level-\(r\) graph \(G=G_r\), and let
\(\Lambda=d_G(s_r,t_r)\) denote its source-to-target distance.
We will construct a map \(\Psi:G\to\R^D\) from the vertices of $G$ to a real vector space of dimension $D = 2r + 1$; then, in the next section, we will define a convex body chasing instance by taking the convex hull of the image of each request in $G$ under this map.

For this reduction to (approximately) preserve the lower bound, the map we develop must satisfy two requirements. First, it must approximately preserve distances, e.g., the length of the mapped offline path cannot grow too much relative to the cost incurred by an online algorithm. Second, since we will be considering requests defined as the convex hull of the mapped vertices of $G$, the online algorithm will now have the flexibility to choose points that are the \emph{weighted averages} of these vertices, rather than the vertices themselves (as was the restriction in the discrete metric). We thus must ensure that convexification does not reduce the online cost too much relative to the offline optimum.
The second requirement is the main difficulty: in general, convexity can make discrete instances much easier, and arbitrarily mapping the full graph $G$ into an $\calO(r)$-dimensional vector space is unlikely to yield a challenging instance.
Instead, we will exploit the recursive structure of the discrete requests to retain the structural difficulty of the instance by mapping each \emph{level} of recursion to its own two coordinates, which will aggregate all copies at that level.

In the rest of this section, all lengths are measured in the full graph $G$.

\paragraph{Branch functions.}

Consider a copy $H$ of a level-$s$ graph $G_s$ embedded in $G$, for $1 \leq s \leq r$; if $s = r$, this is simply the graph $G$ itself. Recall that its local depth function $h_H(v)$ measures how far $v \in H$ is from the source $s_H$. For $v \in H$, let
\[
 g_H(v)=\min\{h_H(v),L_H-h_H(v)\}
\]
be the distance from \(v\) to whichever endpoint of $H$ is closest.
For each branch \(a\in\{L,R\}\), define $H$'s \emph{branch function} as
\begin{equation}
 \phi_{H,a}(v)=
 \begin{cases}
  g_H(v)&\text{if }v\text{ lies on branch }a\text{ of }H,\\
  0&\text{if }v\text{ lies elsewhere}.
 \end{cases}
 \label{eq:individual-hat}
\end{equation}
When $v$ is either the source or target vertex of $H$, $\phi_{H, L}(v) = \phi_{H, R}(v) = 0$, so these maps are well-defined where the two branches meet. $\phi_{H, a}$ is defined to be zero outside of $H$ so that we can aggregate it across different copies at the same recursion level.

\paragraph{Defining mapped coordinates by aggregating branch functions across copies.}

We will define the map $\Psi : G \to \R^D$ by assigning each (nonzero) level of recursion $s = 1, \ldots, r$ two coordinates, each of which will aggregate the branch functions $\phi_{H, a}(v)$ of a vertex $v$ across \emph{all} copies $H$ at level $s$. Specifically, for $s = 1, \ldots, r$ and $a \in \{L, R\}$, define

\begin{equation} \label{eq:aggregation}
 \psi_{s,a}(v) = \sum_{\text{copies $H$ of $G_s$}}\phi_{H,a}(v).
\end{equation}
Then for a vertex $v \in G$, we define the map $\Psi$ as
\begin{equation}
\Psi(v)=\bigl(h(v),(\psi_{s,L}(v),\psi_{s,R}(v))_{s=1}^r\bigr)
 \in\R^D.
 \label{eq:map}
\end{equation}
Its first coordinate is simply the (global) depth of $v$, measuring progress through the entire graph.
At each level $s$, at most one copy contains $v$ unless $v$ is the source or target of a copy at that level, in which case $\psi_{s, a}(v) = 0$. The corresponding pair of coordinates in the output of $\Psi$ records the
branch containing $v$ and its distance to the closest endpoint of that
copy; both entries are zero when no such copy exists. In particular, we have
\(\Psi(s_r)=\mathbf{0}\) and
\(\Psi(t_r)=(\Lambda,0,\ldots,0)\).

For example, when $r=2$, the five coordinates consist of global depth, two
coordinates shared by all level-1 copies, and two for the full
level-2 graph. As requests move from one level-1 copy to the next,
the second and third coordinates describe progress through each copy, increasing and then decreasing as the requests move from the source to the target of each copy.
As we will show below in Corollary~\ref{cor:shared-coordinates}, two consecutive requests will never require us to distinguish nonzero contributions from different copies at the same
level. Moreover, we show in the following lemma that, as a vertex moves, each coordinate of $\Psi$ cannot change more than the corresponding graph distance.

\begin{lemma}[Coordinatewise movement bound]\label[lemma]{lem:expansion}
Every scalar coordinate of \(\Psi\) is \(1\)-Lipschitz in graph
distance. In particular, for every \(u,v\in G\),
\begin{equation}
 \norm{\Psi(u)-\Psi(v)}_\infty\le d_G(u,v).
 \label{eq:expansion}
\end{equation}
\end{lemma}

\begin{proof}[Proof.]
As discussed in Section~\ref{subsec:graphs}, both $h$ and $h_H$ are $1$-Lipschitz in $d_G$, and thus so is $g_H$ (as the minimum of $1$-Lipschitz functions). Extending $g_H$ by zero outside branch $a$ also preserves this bound: any path from that branch to a vertex outside it must pass through $s_H$ or $t_H$, and its length is therefore at least $g_H(v)$ when it starts at $v$. Thus $\phi_{H,a}$ is $1$-Lipschitz in $d_G$. It follows that, for each $s = 1, \ldots, r$, $\psi_{s, a}$ is also $1$-Lipschitz in $d_G$: this is because for any $v \in G$, at most one component in the sum \eqref{eq:aggregation} is active, i.e., nonzero, since no vertex is contained in two copies at the same level except at a shared source and target. Thus $\psi_{s, a}$ can be written as a maximum over $\phi_{H, a}$ for copies $H$ at level $s$, and the maximum preserves the Lipschitz constant.

\end{proof}

\subsection{Consecutive requests}\label{subsec:transition}

Because algorithms for convex body chasing can take advantage of convexity by choosing a weighted average of vertices, we must understand how distances between these weighted averages relate to distances in the original graph. To this end, we will first describe how two consecutive requests in the Section~\ref{subsec:requests} construction relate: in particular, which vertices can occur in these requests. By the inductive construction of requests, any change of vertices is localized to a call at one level and involves at most three vertices, while higher-level waiting vertices remain feasible. In particular, we have the following result.

\begin{lemma}[Structure of consecutive requests]\label[lemma]{lem:one-block}
Consider consecutive requests \(S,S'\). If $S \neq S'$, then there is a decomposition of these requests into a local request change $U, U'$ in a copy $H$ at some level $0 \leq s \leq r$, plus a set $Z$ of waiting vertices retained by the higher-level calls containing $H$, such that
\begin{equation} \label{eq:request_decomposition}
 S=U\cup Z,\qquad S'=U'\cup Z,\qquad
 (U\cup U')\cap Z=\varnothing.
\end{equation}
The local requests \(U,U'\) within $H$ have one of the following forms:
\begin{center}
\begin{tabular}{@{}lll@{}}
\toprule
Local event in \(H\) & \(U\) & \(U'\) \\
\midrule
Traverse a level-zero copy & \(\{s_H\}\) & \(\{t_H\}\) \\
Traverse the two entrance edges & \(\{s_H\}\) & \(\{u_L,u_R\}\) \\
Start Stage 3 & \(\{u,v\}\) & \(\{v\}\) \\
\bottomrule
\end{tabular}
\end{center}
Here \(u_L,u_R\) are the terminal endpoints of the entrance edges, and in the third row, \(u\) denotes the discarded branch's current
vertex and \(v\) is the surviving one.

On the other hand, if $S = S'$, the common request occurs at the beginning or end of a recursive call to a copy $H$. Let $u \in \{s_H, t_H\}$ be the corresponding endpoint that is requested, and let $Z$ be the waiting vertices retained by higher-level calls. Then we can take $U = U' = \{u\}$ for the decomposition \eqref{eq:request_decomposition}.

In any case, the set \(Z\) contains at most one waiting
vertex from each higher level.
\end{lemma}

\begin{proof}[Proof.]
First, suppose $S \neq S'$. Start with the highest-level call in $G$. If both requests occur within a call to a smaller copy, set aside the waiting vertex on the opposite
branch for $Z$, if one is present, and repeat inside that copy. For distinct requests, this process stops at a base (level-$0$) edge, at the entrance edges of
a copy, or at the removal of a branch's current vertex in Stage 3, as these are the three possible changes which do not occur within a recursive call; $H$ is specified as the terminal copy.

On the other hand, if $S = S'$, repetition occurs at a boundary involving the first or last request of a call to a copy $H$: for instance, it can occur between the entrance edge and the first copy along a branch, or between two copies. We can take $U$ to be the corresponding endpoint of $H$, and $Z$ to be all higher-level waiting vertices in $S$.

In any case, following the argument in the proof of Lemma~\ref{lem:offline-metric}, there is at most one waiting vertex in $Z$ from each higher level.
\end{proof}

This structure allows us to obtain the following corollary characterizing the branch aggregates $\psi_{s, a}$ on consecutive requests.

\begin{corollary}\label[corollary]{cor:shared-coordinates}
  Let $S, S'$ be consecutive requests. For each level $s \in \{1, \ldots, r\}$, there exists a level-$s$ copy $H_s$ such that every other level-$s$ copy $H_s'$ has both branch functions equal to zero on $S \cup S'$---that is, such that $\phi_{H_s', a}(v) = 0$ for all $v \in S \cup S'$, $a \in \{L, R\}$, and level-$s$ copies $H_s' \neq H_s$. Consequently,
  $$\psi_{s, a}(v) = \phi_{H_s, a}(v)$$
  for every $v \in S \cup S'$ and $a \in \{L, R\}$.
\end{corollary}

\begin{proof}[Proof.]
  Let $H$, $U$, $U'$, and $Z$ be the copy, local requests, and waiting vertices supplied by Lemma~\ref{lem:one-block}, respectively, and let $\ell \in \{0, \ldots, r\}$ be the level of $H$. Fix $s \in \{1, \ldots, r\}$.

  If $s \geq \ell$, choose $H_s$ to be the level-$s$ copy containing $H$. We must show that any other level-$s$ copy has branch functions equal to zero on any vertex $v \in S \cup S'$.
  First, note that, by construction, a waiting vertex retained at some level \(j\) is an endpoint of a
  level-\((j-1)\) copy, so all branch functions at levels below $j$ vanish at that vertex. As such, any waiting vertices in $Z$ which were retained above level $s$ have all level-$s$ branch functions equal to zero. Then there are two remaining kinds of vertices in the request. First, the local vertices satisfy $U \cup U' \subseteq H \subseteq H_s$; since distinct copies at the same level meet only at endpoints, every level-$s$ copy other than $H_s$ has branch functions equal to zero on these vertices. Second, any waiting vertex $z \in Z$ retained at a level $j \leq s$ is contained in the level-$j$ copy where it is retained; by the recursive construction, that copy contains $H$ and is itself contained in $H_s$, so $z \in H_s$, and we once again have that every level-$s$ copy other than $H_s$ has branch functions equal to zero on $z$.

  On the other hand, if $s < \ell$, the configurations of the local requests $U, U'$ detailed in Lemma~\ref{lem:one-block}---including the case of a repeated request, where $U = U'$--- imply that every local vertex in $U \cup U'$ is either outside of or at an endpoint of every level-$s$ copy, so all level-$s$ branch functions are zero on these vertices. The higher-level waiting vertices likewise have zero level-$s$ branch functions, by the same reasoning as in the previous case. Thus all these functions vanish on $S \cup S'$, and any level-$s$ copy can serve as $H_s$.

  The identity in the corollary statement then follows from the definition of $\psi_{s, a}$.
\end{proof}

Crucially, the copy $H_s$ designated in the above corollary is the \emph{same} for any vertex in either request. Thus, when considering consecutive requests, we can simply restrict our focus to the single copy that realizes the branch aggregate $\psi_{s, a}$ (and any change thereof) between requests.

We will also need to bound the distance between any two vertices that co-occur in the same request in a copy in terms of the branch function.

\begin{lemma}[Distance between co-occurring vertices]\label[lemma]{lem:stationary}
Let $u, v$ be vertices that lie on opposite branches of a copy $H$ and co-occur in a request in a call to $H$. Then
\begin{equation}
 d_G(u, v)\le4g_H(u).
 \label{eq:stationary}
\end{equation}
\end{lemma}

\begin{proof}[Proof.]
Scale \(H\) so that the large copies constituting it have length $1$ (as in the construction in Section~\ref{subsec:graphs}), and write $m, b$ for its construction parameters.
Recall that $L_H = 3m + b$ and $1/m\le b\le1\le m$.

Let \(\delta=1/m\) and \(z=h_H(u)\). In the entrance-pair request and throughout Stage 1, the construction
ensures that \(h_H(v)\le z+\delta\), while
\(\delta\le b\le z\le b+m\le L_H/2\). Thus \(g_H(u)=z\),
and the path through \(s_H\) gives
\[
 d_G(u,v)\le h_H(u)+h_H(v)\le2z+\delta\le3g_H(u).
\]
In Stage 2 and at the start of Stage 3 (before the more advanced vertex is removed), the relevant vertices have
depths in \([b+m,b+2m]\). In particular, \(g_H(u)\ge m\),
whereas \eqref{eq:copy-distance-bound} gives
\(d_G(u,v)\le L_H\le4m\), giving \eqref{eq:stationary}. Rescaling preserves these bounds.
\end{proof}

Having established these structural properties of the discrete requests and the branch functions, we next turn to relating the costs of randomized strategies in the discrete instance and corresponding weighted averages in the mapped instance.

\subsection{Movement between weighted averages}\label{sec:transport}

Let \(P,Q\) be probability distributions on the vertices of \(G\).
A \emph{coupling} \(\pi\) of \(P,Q\) assigns a nonnegative mass
\(\pi(u,v)\) to move from each vertex \(u\) to each vertex \(v\), satisfying the marginal constraints
\(\sum_v\pi(u,v)=P(u)\) and \(\sum_u\pi(u,v)=Q(v)\).
Its cost is \(\sum_{u,v}\pi(u,v)d_G(u,v)\); let \(W_1(P,Q)\)
denote the minimum cost over all such couplings, i.e., the Wasserstein-1 distance between $P$ and $Q$. We use the same definition for any two nonnegative finite measures with equal total mass.

We now show that,
for distributions on consecutive requests, this cost is controlled
by the $\ell_1$ movement of their weighted averages under the map $\Psi$.

\begin{lemma}[Transportation cost is controlled by weighted average movement]\label[lemma]{lem:transport}
Let \(S,S'\) be consecutive requests in $G$, including possible repetitions.
For probability distributions \(P\) on \(S\) and \(Q\) on \(S'\),
with masses extended by zero elsewhere,
\begin{equation}
 W_1(P,Q)\le8\left\|\sum_{v\in G}(Q(v)-P(v))\Psi(v)\right\|_1.
 \label{eq:transport}
\end{equation}
\end{lemma}

\begin{proof}[Proof.]
Let \(\mu=Q-P\) be the signed, discrete measure denoting the change in masses, and for a scalar
function $f : G \to \R$, write \(\mu(f)=\sum_u\mu(u)f(u)\). Negative
entries of \(\mu\) specify mass to remove from a given vertex, and positive entries
specify mass to supply. For any signed, discrete measure $\eta$ of total mass $\eta(1) = 0$, let
\(\TC(\eta)\) be the minimum cost of moving its
negative masses to its positive masses; clearly $\TC(\eta) = W_1(\eta_+, \eta_-)$, where $\eta_+$ and $\eta_-$ are the positive and negative parts of $\eta$, respectively. Thus
\(W_1(P,Q)=\TC(\mu)\). It is easily verified that $\TC$ satisfies the triangle inequality: for signed, discrete measures $\eta, \eta'$ with total mass $0$, \(\TC(\eta+\eta')\le\TC(\eta)+\TC(\eta')\).

Using \cref{lem:one-block}, write \(S\cup S'=E\cup Z\), where
\(E=U\cup U'\) contains at most three local vertices and \(Z\)
consists of waiting vertices. We will account for the cost of the mass change at each waiting vertex $z \in Z$ by a transfer between it and one vertex $v \in E$, which we will choose later.
After these transfers, only mass changes within the local set \(E\) remain to be handled.

Supposing that $v \in E$ has been fixed, and writing \(\delta_u\) for a unit mass at \(u\), define the remaining signed measure \(\nu\) and the transportation cost \(B\) of the transfers between $Z$ and $v$ by
\begin{equation}
 \nu=\mu-\sum_{z\in Z}\mu(z)(\delta_z-\delta_v),
 \qquad
 B=\sum_{z\in Z}|\mu(z)|d_G(z,v).
 \label{eq:peel}
\end{equation}
For each $z \in Z$, subtracting $\mu(z)(\delta_z - \delta_v)$ sets the (signed) mass at $z$ to zero and adds $\mu(z)$ to the (signed) mass at $v$.
Thus \(\nu\) is zero on \(Z\),
agrees with \(\mu\) on \(E\setminus\{v\}\), and has total
mass zero. By the triangle inequality for \(\TC\), the original
transportation cost is at most \(B+\TC(\nu)\). Since the global depth \(h\) is
\(1\)-Lipschitz in graph distance, these transfers also change the weighted depth by
at most \(B\): we have
$$|\nu(h) - \mu(h)| = \left|\sum_{z \in Z} \mu(z)(h(v) - h(z))\right| \leq \sum_{z \in Z}|\mu(z)||h(v) - h(z)| \leq \sum_{z \in Z}|\mu(z)|d_G(z, v) = B.$$
Thus
\begin{equation}
 W_1(P,Q)\le B+\TC(\nu),
 \qquad
 |\nu(h)|\le|\mu(h)|+B.
 \label{eq:depth-error}
\end{equation}

To bound the contribution of each $z$ to $B$, let \(H_z\) be the
copy whose call retains \(z\), let $s_z$ be its level, and let $a_z \in \{L, R\}$ denote its branch containing $z$. Set $q_z \coloneqq \psi_{s_z, a_z}$ to be the corresponding aggregate branch function, which is the $(2s_z + \indic_{a_z = R})$th coordinate of $\Psi$.
This coordinate records only the mass at \(z\) on the two
requests. Indeed, consider the individual branch function $\phi_{H_z, a_z}$. At $z$, it equals $g_{H_z}(z) > 0$, since $z$ (being a waiting vertex) is neither the source nor the target of $H_z$. On the other hand, it vanishes at every other vertex of $S \cup S'$: the local vertices and waiting vertices retained below level $s_z$ lie on the branch of $H_z$ opposite $z$, while waiting vertices above level $s_z$ have zero branch functions at this level. Because $\phi_{H_z, a_z}$ is nonzero at $z$,
\cref{cor:shared-coordinates} implies that no other level-$s_z$ copy contributes to $q_z$ on any vertex in \(S\cup S'\). Thus $q_z = \phi_{H_z, a_z}$ on $S \cup S'$, giving
\begin{equation}
 \mu(q_z)=\mu(z)g_{H_z}(z).
 \label{eq:isolate}
\end{equation}
Since \(z\) belongs to both requests, it co-occurs with every
\(v\in E\) in \(S\) or \(S'\), and every such $v$ lies on the opposite branch of
\(H_z\). Thus \cref{lem:stationary} bounds the transfer cost by
\( |\mu(z)|d_G(z,v)\le4|\mu(q_z)| \). Summing over $Z$ gives
\begin{equation}
 B\le4\sum_{z\in Z}|\mu(q_z)|.
 \label{eq:B}
\end{equation}

It remains to select the designated transfer source/recipient $v \in E$ and bound \(\TC(\nu)\); we consider four cases based on the local configuration of $E$ (specifically, the four cases in Lemma~\ref{lem:one-block}).

\paragraph{Case 1: the local requests traverse a level-zero copy.}
Suppose we are in the first case of the table in Lemma~\ref{lem:one-block}, and let $s_H$ and $t_H$ denote the source and target of the level-zero copy/edge. Choose either endpoint as \(v\). Because $\nu$ is zero on $Z$, we have $\nu(s_H) + \nu(t_H) = 0$. The edge length is exactly the depth difference, so $\TC(\nu) = |\nu(s_H)|d_G(s_H, t_H) = |\nu(h)|$. Thus
\(W_1(P,Q)\le|\mu(h)|+2B\).

\paragraph{Case 2: the local requests traverse the entrance edges.}
Suppose we are in the second case of the table in Lemma~\ref{lem:one-block}, and that the local change occurs in a level-\(s\) copy \(H\)
whose entrance edges have length \(b_H\). Choose \(v=s_H\).
The residual masses at the terminal vertices \(u_L,u_R\) of the entrance edges are still \(\mu(u_L),\mu(u_R)\), and they can be supplied by transporting mass from $s_H$ along the two entrance edges at cost at most
\begin{equation}
 \TC(\nu)\le b_H\bigl(|\mu(u_L)|+|\mu(u_R)|\bigr)
 =|\mu(\psi_{s,L})|+|\mu(\psi_{s,R})|.
 \label{eq:split}
\end{equation}
To see the last equality, note that, by the definition of the aggregate branch function and the reasoning in the proof of Corollary~\ref{cor:shared-coordinates}, for each branch $a \in \{L, R\}$ we have $\psi_{s, a}(u_a) = g_H(u_a) = b_H$, while $\psi_{s, a}$ vanishes at every other vertex in \(S\cup S'\).

\paragraph{Case 3: a removal at the start of Stage 3.}
Suppose we are in the third case of the table in Lemma~\ref{lem:one-block}. Choose the surviving vertex as \(v\), and let \(u\) be the
removed vertex. If the copy $H$ in which the removal occurred has level $s$ and
\(u\) lies on branch \(a\), then
\begin{equation}
 \TC(\nu)=|\mu(u)|d_G(u,v)
 \le4|\mu(u)|g_H(u)=4|\mu(\psi_{s,a})|,
 \label{eq:delete}
\end{equation}
by \cref{lem:stationary}. The last equality follows because $\psi_{s, a}(u) = g_H(u)$, while \(\psi_{s,a}\) vanishes at the surviving vertex (which is on the opposite branch from $u$) and at all vertices in \(Z\).

\paragraph{Case 4: a repeated request.}
Suppose we are in the last case (not in the table) in Lemma~\ref{lem:one-block}, where a request is repeated. Here $E$ is a singleton; let $v$ be its unique vertex. Since $\nu$ is supported on $\{v\}$ and has total mass zero, \(\nu=0\), and thus \(W_1(P,Q)\le B\).

\medskip

By \cref{lem:one-block}, there is at most one waiting vertex
at each level, so the \(q_z\) correspond to distinct coordinates of
\(\Psi\) for each $z \in Z$. Thus \eqref{eq:B} gives
\[
 B\le4\sum_{s=1}^r\sum_{a\in\{L,R\}}
 |\mu(\psi_{s,a})|.
\]
In the level-zero edge case, we have
\(W_1(P,Q)\le|\mu(h)|+2B\).
In the other three cases, the bounds above give
\[
 \TC(\nu)\le4\sum_{s=1}^r\sum_{a\in\{L,R\}}
 |\mu(\psi_{s,a})|.
\]
Together with \(W_1(P,Q)\le B+\TC(\nu)\), these bounds imply
that, in every case,
\[
 W_1(P,Q)
 \le |\mu(h)|+
 8\sum_{s=1}^r\sum_{a\in\{L,R\}}|\mu(\psi_{s,a})|
 \le8\left\|\sum_u\mu(u)\Psi(u)\right\|_1,
\]
where the last inequality follows from the definition of
the coordinates of \(\Psi\).
\end{proof}

\section{The family of convex body chasing instances}\label{sec:convexification}

We will now describe and analyze the construction giving the convex body chasing lower bound. We will first analyze the construction in the case where movement costs are $\ell_p$ (Proposition~\ref{prop:body}); then, we will use the Euclidean case ($p = 2$) to establish the result for an arbitrary norm.

Fix $p \in [1, \infty]$. For each discrete request \(S_t\) in the metrical service systems instance, define
\begin{equation}
 K_t=\conv\{\Psi(v):v\in S_t\}\subseteq\R^D,
 \label{eq:convexification}
\end{equation}
with initial point \(\mathbf{x}_0=\Psi(s_r)=\mathbf{0}\).
These sets are nonempty and compact. The graph, the map, and the
random choices generating the requests are independent of any
algorithm, so this defines an oblivious distribution of convex body chasing
instances. In what follows, it suffices to restrict our focus to deterministic algorithms, since any randomized algorithm for convex body chasing can be derandomized by taking the expectation while preserving feasibility and without increasing movement cost \cite{bansal_2-competitive_2015}.

We will compare the online and offline costs on these instances using equivalence of norms: for $\mathbf{u}\in\R^D$,
\begin{equation}
 \norm{\mathbf{u}}_p\le D^{1/p}\norm{\mathbf{u}}_\infty,
 \qquad
 \norm{\mathbf{u}}_1\le D^{1-1/p}\norm{\mathbf{u}}_p.
 \label{eq:norm-comparisons}
\end{equation}
The first inequality will be used to bound the cost of the mapped offline sequence. The second, together with \cref{lem:transport}, lets us construct a randomized online algorithm for the discrete requests whose expected movement cost is at most $8D^{1-1/p}$ times the convex body chasing algorithm's expected movement cost. In particular, we have the following result.

\begin{proposition}\label[proposition]{prop:body}
Let movement costs be $\ell_p$. There is a universal constant \(c_0>0\) such that, for \(r\ge1\), any convex body chasing algorithm $\alg$ run on the request sequence in dimension $D = 2r + 1$ satisfies
$$\E[\alg] \geq \frac{c_0r^2\Lambda}{8D^{1-1/p}},$$
where $\Lambda = d_G(s_r,t_r)$ is the source-to-target distance of the graph $G$ used to construct the instance, and algorithm movement cost is measured in the $\ell_p$ norm. The expectation is over the request sequence drawn from
\(\mathcal D_r\), which gives a correspondingly random sequence of convex body requests.

On the other hand, every realized sequence of bodies satisfies
\[
 \opt \le D^{1/p}\Lambda.
\]
\end{proposition}

\begin{proof}[Proof.]
Let $\alg$ be an arbitrary (deterministic) algorithm and let \(\mathbf{x}_t\in K_t\) be its decisions at each time.
Choose a probability distribution \(P_t\) on \(S_t\) representing
each point:
\[
 \mathbf{x}_t=\sum_{v\in S_t}P_t(v)\Psi(v).
\]
Such weights exist by the definition of a convex hull; if there are multiple feasible choices of $P_t$, break ties using a fixed rule.

We will turn these distributions into an algorithm $\alg_{\text{discrete}}$ for the discrete request sequence in $G$. Set \(P_0=\delta_{s_r}\). At each time
\(t\), choose a minimum-cost coupling \(\pi_t\) of
\(P_{t-1},P_t\), using a deterministic rule to resolve ties.
If the previously selected vertex is \(u\), at time $t$ select \(v\)
with probability \(\pi_t(u,v)/P_{t-1}(u)\); vertices of
zero \(P_{t-1}\)-mass are never selected. These choices
use only the requests revealed so far, so this defines
a randomized online algorithm for the discrete requests.

For each fixed input, induction shows that consecutive
selected vertices have joint distribution \(\pi_t\).
Their expected movement is therefore \(W_1(P_{t-1},P_t)\).
The first request is the source singleton, so the initial
movement is zero. At each later step, \cref{lem:transport}
and \eqref{eq:norm-comparisons} give
\[
 W_1(P_{t-1},P_t)
 \le8\norm{\mathbf{x}_t-\mathbf{x}_{t-1}}_1
 \le8D^{1-1/p}\norm{\mathbf{x}_t-\mathbf{x}_{t-1}}_p.
\]
Summing over time, averaging over the distribution on request sequences, and applying the distributional lower bound from \cref{lem:mss-lower} to the cost of the discrete algorithm we have specified yields
\[
 c_0r^2\Lambda
 \le\E[\alg_{\mathrm{discrete}}]
 \le8D^{1-1/p}\E[\alg],
\]
for a
universal \(c_0>0\).

For the offline comparison, \cref{lem:offline-metric} gives a
feasible directed path of length \(\Lambda\) in every realized
discrete instance. By \cref{lem:expansion} and
\eqref{eq:norm-comparisons},
\[
 \norm{\Psi(u)-\Psi(v)}_p\le D^{1/p}d_G(u,v).
\]
Summing over time, the mapped offline choices in the discrete instance therefore serve the convex requests at
total cost at most \(D^{1/p}\Lambda\), giving the upper bound on $\opt$.
\end{proof}

The online and offline comparisons together lose a combined factor of
\(8D\), independent of \(p\). Since \(D=2r+1\), the quadratic
discrete-set lower bound gives the desired linear lower bound for convex body chasing with $\ell_p$ movement costs. To extend the result to arbitrary norms, we will use the following consequence of Milman's quotient-of-a-subspace theorem \cite{milman_1985} (see also \cite[Section 2.4]{Giannopoulos2004}).

\begin{lemma}[Milman~\cite{milman_1985}]
\label[lemma]{lem:milman}
There is a universal constant \(M\ge1\) such that, for every
\(d\)-dimensional normed vector space \(X\), there is a subspace
\(E\subseteq X\), an integer \(k\ge d/2\), and a linear map
\(Q:E\to\R^k\) satisfying
\[
  B_2^k\subseteq Q(B_E)\subseteq M \cdot B_2^k,
\]
where \(B_E=\{\mathbf{x}\in E:\norm{\mathbf{x}}_X\le1\}\)
is the unit ball of \(E\), and \(B_2^k\) is the Euclidean
unit ball in $\R^k$.
\end{lemma}

The first inclusion allows each Euclidean movement to be lifted
to \(E\) at no greater cost, while the second bounds the cost
of projecting a trajectory from \(E\) into \(\R^k\). With this result in place, we can now prove the general lower bound.

\medskip

\begin{proof}[Proof of Theorem~\ref{thm:main}.]
Fix a normed vector space \(X=(\R^d,\norm{\cdot}_X)\) with \(d\ge6\),
and choose the subspace $E$, quotient dimension $k$, and linear map $Q$ as in \cref{lem:milman}.
Set \(r=\lfloor(k-1)/2\rfloor\) and \(D=2r+1\).

We will apply the $p = 2$ case of \cref{prop:body}, appending zero coordinates
to its requests \(K_t\) as needed to view them as subsets of \(\R^k\).
Writing \(B=\sqrt D\,\Lambda\), the proposition tells us that every realized instance
has Euclidean offline optimum $\opt_2 \leq B$, while every deterministic online algorithm has Euclidean cost \(\alg_2\) satisfying
\begin{equation} \label{eq:euclidean_lb}
 \E[\alg_2]\ge\frac{c_0r^2\Lambda}{8\sqrt D}.
\end{equation}
For each Euclidean request \(K_t\), we will define a transformed request $\widetilde K_t$ in \(X\) by taking the intersection of $K_t$'s preimage under $Q$ with the $X$-ball of radius $B$:
\[
  \widetilde K_t
  =\{\mathbf{x}\in E:
      \norm{\mathbf{x}}_X\le B,\ Q\mathbf{x}\in K_t\}.
\]
Each $\widetilde K_t$ is clearly closed and convex; the bound on \(\norm{\mathbf{x}}_X\) also makes these requests compact, so the lower bound holds even under a restriction to compact requests. To see that these requests preserve the fixed upper bound $B$ on the offline optimal cost, fix a realized sequence of requests and a feasible Euclidean trajectory \(\mathbf{u}_0=\mathbf{0},\mathbf{u}_t\in K_t\) with cost bound $\sum_{t=1}^T \|\mathbf{u}_t - \mathbf{u}_{t-1}\|_2 \leq B$. Because \(B_2^k\subseteq Q(B_E)\), we can choose \(\mathbf{v}_t\in \norm{\mathbf{u}_t - \mathbf{u}_{t-1}}_2 \cdot B_E\) such that
\[
  Q\mathbf{v}_t=\mathbf{u}_t-\mathbf{u}_{t-1}
\]
and
$$
  \norm{\mathbf{v}_t}_X
    \le\norm{\mathbf{u}_t-\mathbf{u}_{t-1}}_2.$$
Set \(\mathbf{x}_0=\mathbf{0}\) and define
\(\mathbf{x}_t=\sum_{j=1}^t\mathbf{v}_j\). Then
\(Q\mathbf{x}_t=\mathbf{u}_t\), and we have both
\[
  \norm{\mathbf{x}_t}_X
  \le\sum_{j=1}^t\norm{\mathbf{v}_j}_X
  \le\sum_{j=1}^t
       \norm{\mathbf{u}_j-\mathbf{u}_{j-1}}_2
  \le B,
\]
and
$$\sum_{t=1}^T \norm{\mathbf{x}_t - \mathbf{x}_{t-1}}_X = \sum_{t=1}^T \norm{\mathbf{v}_t}_X \leq \sum_{t=1}^T \norm{\mathbf{u}_t - \mathbf{u}_{t-1}}_2 \leq B.$$
Thus each \(\mathbf{x}_t\in\widetilde K_t\) and the total cost (in $\|\cdot\|_X$) of following this sequence is at most $B$, so we have the identical optimal cost bound \(\opt_X\le B\).

Now fix an online algorithm \(\alg_X\) for the transformed requests. We can turn this into a Euclidean online algorithm for the original sequence $K_1, K_2, \ldots$ by transforming each arriving \(K_t\) into the corresponding \(\widetilde K_t\),
running \(\alg_X\) to obtain a decision $\mathbf{x}_t \in E$, and returning \(Q\mathbf{x}_t\) as the Euclidean decision for that round. By construction of the transformed request sequence, $Q\mathbf{x}_t \in K_t$, and the resulting algorithm is online because the transformation uses only $K_t$ and the fixed objects $E, Q$, and $B$.

Since \(Q(B_E)\subseteq M \cdot B_2^k\), mapping the trajectory $\mathbf{x}_t$ from $E$ to $\R^k$ (under $Q$) will increase its movement cost by at most the constant factor $M$. Thus, if \(\alg_X\) has competitive ratio \(R\) for instances in $X$, \eqref{eq:euclidean_lb} gives
\[
  \frac{c_0r^2\Lambda}{8\sqrt D}
  \le\E[\alg_2]
  \le M \cdot \E[\alg_X]
  \le MR \cdot \,\E[\opt_X]
  \le MR\sqrt D\,\Lambda.
\]
Rearranging gives
\begin{equation} \nonumber
  R\ge\frac{c_0r^2}{8MD}=\Omega(d),
  \label{eq:final-ratio}
\end{equation}
since \(r\ge k/4\), \(D\le k\), and \(k\ge d/2\).
In particular, the constant is universal and independent of the norm.
For \(2\le d\le5\), the claim follows from $R \geq 1$ and reducing the constant if needed.
\end{proof}

\section{Acknowledgments} \label{sec:ack}

This work was supported in part by a grant of access to OpenAI models through the ChatGPT for Academic Researchers program.

\paragraph{AI use statement.}
The lower bound in this paper was obtained via an extended interaction with OpenAI's GPT-6 Pro model; in particular, the use of the \cite{bubeck_randomized_2022} construction and its mapping to convex requests was the suggestion of this model. The author had previously attempted to obtain such a lower bound using various AI models (including GPT-6 Pro and previous versions of GPT Pro) without success. GPT-6 Pro also produced the first draft of this manuscript, which the author has substantially revised to improve clarity and accessibility. The author has verified the correctness and originality of all content, and takes full responsibility for the manuscript and its results.

\bibliographystyle{alpha}
\bibliography{references}

\end{document}